\documentclass[11pt]{article}%
\usepackage{amsfonts}
\usepackage{amsmath}
\usepackage{amssymb}
\usepackage{graphicx}%
\providecommand{\U}[1]{\protect\rule{.1in}{.1in}}
\newtheorem{theorem}{Theorem}
\newtheorem{acknowledgement}[theorem]{Acknowledgement}

\newtheorem{corollary}[theorem]{Corollary}

\newtheorem{lemma}[theorem]{Lemma}

\newtheorem{proposition}[theorem]{Proposition}

\newenvironment{proof}[1][Proof]{\noindent\textbf{#1.} }{\ \rule{0.5em}{0.5em}}
\begin{document}

\title{Generalized Anti-Wick Quantum States }
\author{Maurice de Gosson\thanks{maurice.de.gosson@univie.ac.at}\\University of Vienna\\Faculty of Mathematics, NuHAG\\and\\University of W\"{u}rzburg\\Institute of Mathematics}
\maketitle

\begin{abstract}
The purpose of this paper is to study a simple class of mixed states and the
corresponding density operators (matrices). These operators, which we call
quite Toeplitz density operators correspond to states obtained from a fixed
function (\textquotedblleft window\textquotedblright) by position-momentum
translations, and reduce in the simplest case to the anti-Wick operators
considered long ago by Berezin. The rigorous study of Toeplitz operators
requires the use of classes of functional spaces defined by Feichtinger.

\end{abstract}

\section{Introduction}

\subsubsection*{Motivations}

A quantum mixed state on $\mathbb{R}^{n}$ consists of a collection of pairs
$\{(\psi_{j},\alpha_{j})\}$ where the $\psi_{j}\in L^{2}(\mathbb{R}^{n})$ are
\textquotedblleft pure states\textquotedblright\ and the $\alpha_{j}$ are
probabilities summing up to one. Given a mixed state one defines its density
operator as the bounded operator $\widehat{\rho}=\sum_{j}\alpha_{j}%
\widehat{\Pi}_{j}$ on $L^{2}(\mathbb{R}^{n})$ where $\widehat{\Pi}_{j}$ is the
orthogonal projection in $L^{2}(\mathbb{R}^{n})$ onto the ray $\mathbb{C}%
\psi_{j}=\{\alpha\psi_{j}:\alpha\in\mathbb{C}\}$. The corresponding Wigner
distribution is by definition the function $\rho\in L^{2}(\mathbb{R}^{2n})$
defined by
\begin{equation}
\rho(z)=\sum_{j}\alpha_{j}W\psi_{j}(z) \label{ro1}%
\end{equation}
here $W\psi_{j}$ is the usual Wigner transform of $\psi_{j}$. In this work we
will deal with a particular class of mixed states, and their generalization to
a continuous setting. Consider the ground state
\begin{equation}
\phi_{z_{0}}(x)=(\pi\hbar)^{-n/4}e^{ip_{0}(x-x_{0})}e^{-|x-x_{0}|^{2}/2\hbar}
\label{pho1}%
\end{equation}
of a linear oscillator $H=\frac{1}{2}(|p-p_{0}|^{2}+|x-x_{0}|^{2})$ whose
center $z_{0}=(x_{0},p_{0})$ is not know with precision. Suppose first we have
some partial information telling us that the center $z_{0}$ is located
somewhere on a lattice $\Lambda\subset\mathbb{R}^{2n}$ consisting of a
discrete set of phase space points $z_{\lambda}=(x_{\lambda},p_{\lambda})$ and
that there is a probability $\mu_{z_{\lambda}}$ that the system is precisely
in the state $\phi_{z_{\lambda}}$. Denoting by $\widehat{\Pi}_{z_{\lambda}}$
the orthogonal projection on $\mathbb{C}\phi_{z_{\lambda}}$ the corresponding
density operator
\begin{equation}
\widehat{\rho}=\sum_{z_{\lambda}\in\Lambda}\mu_{z_{\lambda}}\widehat{\Pi
}_{z_{\lambda}} \label{ro2}%
\end{equation}
has Wigner distribution%
\begin{equation}
\rho(z)=\sum_{z_{\lambda}\in\Lambda}\mu_{z_{\lambda}}W\phi_{z_{\lambda}}(z)~.
\label{ro3}%
\end{equation}
Observing that $W\phi_{z_{\lambda}}(z)=W\phi_{0}(z-z_{\lambda})$ this reduces
to the Gabor-type \cite{Faul} expansion
\begin{equation}
\rho(z)=\sum_{z_{\lambda}\in\Lambda}\mu_{z_{\lambda}}W\phi_{0}(z-z_{\lambda
})~. \label{Gabor1}%
\end{equation}
Let us now depart from the discrete case, and consider the somewhat more
realistic situation where the center of the linear oscillator can be any point
$z^{\prime}$ in phase space $\mathbb{R}^{2n}$; we assume the latter comes
equipped with a certain Borel probability density $\mu$. Formula
(\ref{Gabor1}) suggests that we define in this case a generalized Wigner
distribution by
\begin{equation}
\rho(z)=\int\mu(z^{\prime})W\phi_{0}(z-z^{^{\prime}})d^{2n}z^{\prime}=(\mu\ast
W\phi_{0})(z)~. \label{Gabor2}%
\end{equation}
Notice that if one chooses for $\mu$ the atomic measure $\sum_{z_{\lambda}%
\in\Lambda}\mu_{z_{\lambda}}\delta_{z_{\lambda}}$ then (\ref{Gabor2}) reduces
to the discrete sum (\ref{ro3}). It turns out that the operator $\widehat{\rho
}$ obtained from $\rho$ using the Weyl correspondence is a well-known
mathematical object: it is the \textit{anti-Wick operator} with Weyl symbol
$\rho$. Such operators were first considered by Berezin \cite{Berezin}, and
have been developed since by many independent authors. A further
generalization of (\ref{Gabor2}) now consists in replacing the standard
Gaussian $\phi_{0}$ by an arbitrary square integrable function $\phi$ and to
consider the Weyl transforms of functions of the type
\begin{equation}
\rho(z)=\int\mu(z^{\prime})W\phi(z-z^{^{\prime}})d^{2n}z^{\prime}=(\mu\ast
W\phi)(z)~. \label{Gabor3}%
\end{equation}
The operators thus obtained are called \textit{Toeplitz operators} (or
\textit{localization operators}) in the mathematical literature; they are
natural generalizations of anti-Wick operators \cite{Englis,grotoft,sh87}. So
far, so good. The rub comes from the fact that it is not clear why the Weyl
transform $\widehat{\rho}$ of (\ref{Gabor2}) --- let alone that of
(\ref{Gabor3})! --- should indeed be density operators. For this,
$\widehat{\rho}$ has to satisfy three stringent conditions: \textit{(i)}
$\widehat{\rho}$ must be positive semi-definite: $\widehat{\rho}\geq0$ and
\textit{(ii)} be self-adjoint: $\widehat{\rho}=\widehat{\rho}^{\ast}$;
\textit{(iii)} $\widehat{\rho}$ must be of trace class and have trace one:
$\operatorname*{Tr}(\widehat{\rho})=1$. While it is easy to verify (i) (which
implies (ii) since $L^{2}(\mathbb{R}^{n})$ is complex), it is the third
condition which poses problem since it is certainly not trivially satisfied,
as we will see in the course of this paper. Viewed in a broader perspective,
quasi-distributions of the type $\rho=\mu\ast W\phi$ belong to the Cohen class
\cite{Birkbis,Gro}, which has a rich internal structure and is being currently
very much investigated for its own sake both in time-frequency analysis
\cite{boco1,boco,bocogr,cogr,coro,LuSk1} and in quantum mechanics \cite{BJ}.

\subsubsection*{Main results}

\subsubsection*{Notation and terminology}

The scalar product on $L^{2}(R^{n})$ is defined by%
\begin{equation}
(\psi|\phi)_{L^{2}}=\int\psi(x)\phi^{\ast}(x)d^{n}x \label{L2}%
\end{equation}
and we thus have $(\psi|\phi)_{L^{2}}=\langle\phi|\psi\rangle$ in Dirac
bra-ket notation; in this notation $\widehat{\Pi}_{\phi}=|\phi\rangle
\langle\phi|$. The phase space $T^{\ast}R^{n}\equiv R^{2n}$ will be equipped
with the canonical symplectic structure $\sigma=\sum_{j=1}^{n}dp_{j}\wedge
dx_{j}$, given in matrix notation by $\sigma(z,z^{\prime})=Jz\cdot z^{\prime}$
where $J=%
\begin{pmatrix}
0 & I\\
-I & 0
\end{pmatrix}
$ is the standard symplectic matrix on $R^{2n}$. We denote by
\begin{equation}
\widehat{T}(z)=\widehat{T}(x,p)=e^{-i(x\widehat{p}-p\widehat{x})/\hbar}
\label{hw}%
\end{equation}
the Heisenberg--Weyl displacement operator. If $a\in\mathcal{S}^{\prime
}(\mathbb{R}^{2n})$ is a symbol on $\mathbb{R}^{2n}$ the corresponding Weyl
operator is
\begin{equation}
\widehat{A}_{\mathrm{W}}=\operatorname*{Op}\nolimits_{\mathrm{W}}(a)=\int
a_{\sigma}(z)\widehat{T}(z)d^{2n}z \label{aweyl}%
\end{equation}
where $a_{\sigma}$ is the symplectic Fourier transform of $a$, formally
defined by%
\begin{equation}
a_{\sigma}(z)=\left(  \tfrac{1}{2\pi\hbar}\right)  ^{n}\int e^{-\frac{i}%
{\hbar}\sigma(z,z^{\prime})}a(z^{\prime})d^{2n}z^{\prime}~. \label{sft}%
\end{equation}
The Wigner transform of $\psi\in L^{2}(\mathbb{R}^{n})$ is defined by
\begin{equation}
W\psi(z)=\left(  \tfrac{1}{2\pi\hbar}\right)  ^{n}\int e^{-\tfrac{i}{\hbar}%
py}\psi(x+\tfrac{1}{2}y)\psi^{\ast}(x-\tfrac{1}{2}y)d^{n}y~. \label{wt}%
\end{equation}
Similarly, the cross-Wigner transform of two functions $\psi,\phi\in
L^{2}(\mathbb{R}^{n})$ is defined by
\begin{equation}
W(\psi,\phi)(z)=\left(  \tfrac{1}{2\pi\hbar}\right)  ^{n}\int e^{-\tfrac
{i}{\hbar}py}\psi(x+\tfrac{1}{2}y)\phi^{\ast}(x-\tfrac{1}{2}y)d^{n}y~.
\label{crosswt}%
\end{equation}
We have
\begin{equation}
\int W(\psi,\phi)(z)d^{2n}z=(\psi|\phi)_{L^{2}}~. \label{marge}%
\end{equation}

\section{Density Operators: Summary}

A density operator on $L^{2}(\mathbb{R}^{n})$ is by definition a bounded
operator $\widehat{\rho}$ which is semidefinite positive: $\widehat{\rho}%
\geq0$ (and hence self-adjoint), and has trace $\operatorname*{Tr}%
(\widehat{\rho})=1$. The set $\mathcal{D}=\mathcal{D}(L^{2}(\mathbb{R}^{n}))$
of density operators is a convex subset of the space $\mathcal{L}%
_{1}=\mathcal{L}_{1}(L^{2}(\mathbb{R}^{n}))$ of trace class operators (recall
that the latter is a two-sided ideal of the algebra $\mathcal{B}%
=\mathcal{B}(L^{2}(\mathbb{R}^{n}))$ of bounded operators on $L^{2}%
(\mathbb{R}^{n})$). In particular density operators are compact operators, and
hence, by the spectral theorem, there exists an orthonormal basis $(\phi
_{j})_{j}$ of $L^{2}(\mathbb{R}^{n})$ and coefficients satisfying $\lambda
_{j}\geq0$ and $\sum_{j}\lambda_{j}=1$ such that $\widehat{\rho}$ can be
written as a convex sum $\sum_{j}\lambda_{j}\widehat{\Pi}_{\phi_{j}}$ of
orthogonal projections $\widehat{\Pi}_{\phi_{j}}$ converging in the strong
operator topology, that is, counting the multiplicities,
\begin{equation}
\widehat{\rho}\psi=\sum_{j}\lambda_{j}(\psi|\phi_{j})_{L^{2}}\phi_{j}~.
\label{density}%
\end{equation}
By definition, the \textquotedblleft Wigner distribution\textquotedblright\ of
$\widehat{\rho}$ is the function
\begin{equation}
\rho=\sum_{j}\lambda_{j}W\phi_{j}\in L^{2}(\mathbb{R}^{2n})\cap L^{\infty
}(\mathbb{R}^{2n}) \label{rhowig}%
\end{equation}
where $W\psi_{j}$ is the usual Wigner transform of $\phi_{j}$, defined for
$\psi\in L^{2}(\mathbb{R}^{n})$ by
\begin{equation}
W\psi(z)=\left(  \tfrac{1}{2\pi\hbar}\right)  ^{n}\int e^{-\tfrac{i}{\hbar}%
py}\psi(x+\tfrac{1}{2}y)\psi^{\ast}(x-\tfrac{1}{2}y)d^{n}y~. \label{fouwig}%
\end{equation}

\begin{proposition}
\label{Prop1}The Wigner distribution (\ref{rhowig}) is $(2\pi\hbar)^{-n}$
times the Weyl symbol of $\widehat{\rho}$:%
\begin{equation}
\widehat{\rho}=(2\pi\hbar)^{n}\operatorname*{Op}\nolimits_{\mathrm{W}}(\rho).
\label{rhoweyl}%
\end{equation}

\end{proposition}

The proof of this result immediately follows from:

\begin{lemma}
\label{Lemma1}The orthogonal projection $\widehat{\Pi}_{\phi}$ of
$L^{2}(\mathbb{R}^{n})$ on $\mathbb{C}\phi$ ($\phi\in L^{2}(\mathbb{R}^{n})$)
has Weyl symbol $(2\pi\hbar)^{n}W\phi$.
\end{lemma}

\begin{proof}
By definition $\widehat{\Pi}_{\phi}\psi=(\psi|\phi)_{L^{2}}\phi$ so the
distributional kernel of $\widehat{\Pi}_{\phi}$ is the function $K_{\phi
}(x,y)=\phi(x)\phi^{\ast}(y)$; the Weyl symbol of an operator $\widehat{A}%
=\operatorname*{Op}\nolimits_{\mathrm{W}}(a)$ being related to the kernel $K$
of $\widehat{A}$ by the formula \cite{Birkbis,sh87}%
\[
a(z)=\int e^{-\tfrac{i}{\hbar}py}K(x+\tfrac{1}{2}y,x-\tfrac{1}{2}y)d^{n}y
\]
the lemma follows.
\end{proof}

A class of windows $\phi$ playing a privileged role in the study of density
operators is the Feichtinger algebra \cite{Hans1,Hans2}, which is the simplest
modulation space (a non-exhaustive list of references on the topic of
modulation spaces is \cite{Hans0,Gro,Jakob,Toft}. In \cite{Birkbis}, Chapters
16 and 17 we have given a succinct account of the theory using the Wigner
transform instead of the traditional short-time Fourier transforms approach).
By definition Feichtinger's algebra $M^{1}(\mathbb{R}^{n})=S_{0}%
(\mathbb{R}^{n})$ consists of all distributions $\psi\in\mathcal{S}^{\prime
}(\mathbb{R}^{n})$ such that $W(\psi,\phi)\in L^{1}(\mathbb{R}^{2n})$ for
\textit{some} window $\phi\in\mathcal{S}(\mathbb{R}^{n})$; when this condition
holds, we have $W(\psi,\phi)\in L^{1}(\mathbb{R}^{2n})$ for \textit{all}
windows $\phi\in\mathcal{S}(\mathbb{R}^{n})$ and the formula
\begin{equation}
||\psi||_{\phi}=\int|W(\psi,\phi)(z)|d^{2n}z=||W(\psi,\phi)||_{L^{1}} \tag{A1}%
\end{equation}
defines a norm on the vector space $M^{1}(\mathbb{R}^{n})$; another choice of
window $\phi^{\prime}$ the leads to an equivalent norm and one shows that
$M^{1}(\mathbb{R}^{n})$ is a Banach space for the topology thus defined. We
have the following continuous inclusions:
\begin{equation}
\mathcal{S}(\mathbb{R}^{n})\subset M^{1}(\mathbb{R}^{n})\subset C^{0}%
(\mathbb{R}^{n})\cap L^{1}(\mathbb{R}^{n})\cap FL^{1}(\mathbb{R}^{n})\cap
L^{2}(\mathbb{R}^{n}) \tag{A2}\label{inclusions}%
\end{equation}
and $\mathcal{S}(\mathbb{R}^{n})$ is dense in $M^{1}(\mathbb{R}^{n})$.
Moreover, for every $\psi\in L^{2}(\mathbb{R}^{n})$ we have the equivalence
\begin{equation}
\psi\in M^{1}(\mathbb{R}^{n})\Longleftrightarrow W\psi\in L^{1}(\mathbb{R}%
^{n})~. \tag{A3}%
\end{equation}

A particularly important property of $M^{1}(\mathbb{R}^{n})$ is its
metaplectic invariance. Let $\widehat{S}\in\operatorname*{Mp}(n)$ (the
metaplectic group) cover $S\in\operatorname*{Sp}(n)$; then $W(\widehat{S}%
\psi,\widehat{S}\phi)=W(\psi,\phi)\circ S^{-1}$ \cite{Birkbis}; it follows
from this covariance formula and the fact that the choice of $\phi$ is
irrelevant, that $\widehat{S}\psi\in M^{1}(\mathbb{R}^{n})$ if and only if
$\psi\in M^{1}(\mathbb{R}^{n})$. Also, $M^{1}(\mathbb{R}^{n})$ is invariant
under the action of the translations $\widehat{T}(z)$. One can show \cite{Gro}
that $M^{1}(\mathbb{R}^{n})$ is the smallest Banach algebra containing
$\mathcal{S}(\mathbb{R}^{n})$ and having metaplectic and translational invariance.

It turns out that $M^{1}(\mathbb{R}^{n})$ is in addition an algebra for both
pointwise product and convolution; in fact if $\psi,\psi^{\prime}\in
M^{1}(\mathbb{R}^{n})$ then $||\psi\ast\psi^{\prime}||_{\phi}\leq
||\psi||_{L^{1}}||\psi^{\prime}||_{\phi}$ so we have $L^{1}(\mathbb{R}%
^{n})\ast M^{1}(\mathbb{R}^{n})\subset M^{1}(\mathbb{R}^{n})$ hence in
particular
\[
M^{1}(\mathbb{R}^{n})\ast M^{1}(\mathbb{R}^{n})\subset M^{1}(\mathbb{R}%
^{n})~.
\]
Taking Fourier transforms we conclude that $M^{1}(\mathbb{R}^{n})$ is also
closed under pointwise product.

\section{Toeplitz Operators: Definitions and Properties}

Let $\phi$ (hereafter called \emph{window}) be in $M^{1}(\mathbb{R}^{n})$. By
definition, the Toeplitz operator $\widehat{A}_{\phi}=\operatorname*{Op}%
\nolimits_{\phi}(a)$ with window $\phi$ and symbol $a$ is%
\begin{equation}
\widehat{A}_{\phi}=\int a(z)\widehat{\Pi}_{\phi}(z)d^{2n}z \label{toeplitz1}%
\end{equation}
where $\widehat{\Pi}_{\phi}:L^{2}(\mathbb{R}^{n})\longrightarrow
L^{2}(\mathbb{R}^{n})$ is the orthogonal projection onto $\widehat{T}(z)\phi$,
that is%
\begin{equation}
\widehat{\Pi}_{\phi}(z)\psi=(\psi|\widehat{T}(z)\phi)_{L^{2}}\widehat{T}%
(z)\phi~. \label{piffi}%
\end{equation}
We observe that $(\psi|\widehat{T}(z)\phi)_{L^{2}}$ is, up to a factor, the
cross-ambiguity transform of the pair $(\psi,\phi)$; in fact (\cite{Birkbis},
\S 11.4.1)
\begin{equation}
(\psi|\widehat{T}(z)\phi)_{L^{2}}=(2\pi\hbar)^{n}\operatorname{Amb}(\psi
,\phi)(z) \label{psid}%
\end{equation}
where
\begin{equation}
\operatorname{Amb}(\psi,\phi)(z)=\left(  \tfrac{1}{2\pi\hbar}\right)  ^{n}\int
e^{-\tfrac{i}{\hbar}py}\psi(y+\tfrac{1}{2}x)\phi(y-\tfrac{1}{2}x)^{\ast}%
d^{n}y~. \label{crossamb}%
\end{equation}
We can therefore rewrite the definition (\ref{toeplitz1}) of $\widehat{A}%
_{\phi}$ as
\begin{equation}
\widehat{A}_{\phi}\psi=(2\pi\hbar)^{n}\int a(z)\operatorname{Amb}(\psi
,\phi)(z)\widehat{T}(z)\phi d^{2n}z \label{toeplitz2}%
\end{equation}
which is essentially the definition of single-windowed Toeplitz operators
given in the time-frequency analysis literature (see \textit{e.g.}
\cite{cogr,coro,cogr0,Toft}). Toeplitz operators are linear continuous
operators $\mathcal{S}(\mathbb{R}^{n})\longrightarrow\mathcal{S}^{\prime
}(\mathbb{R}^{n}$); in view of Schwartz's kernel theorem they are thus
automatically Weyl operators. In fact:

\begin{proposition}
\label{Prop3}The Toeplitz operator $\widehat{A}_{\phi}=\operatorname*{Op}%
\nolimits_{\phi}(a)$ has Weyl symbol $(2\pi\hbar)^{n}(a\ast W\phi)$, that is
\begin{equation}
\widehat{A}_{\phi}=(2\pi\hbar)^{n}\operatorname*{Op}\nolimits_{\mathrm{W}%
}(a\ast W\phi)~. \label{toeplitz3}%
\end{equation}

\end{proposition}

\begin{proof}
It is sufficient to assume that $a\in\mathcal{S}(\mathbb{R}^{2n})$. Let
$\pi_{\phi}(z)$ be the Weyl symbol of the orthogonal projection $\widehat{\Pi
}_{\phi}(z)$ on $\widehat{T}(z)\phi$; we thus have%
\begin{equation}
(\widehat{\Pi}_{\phi}(z)\psi|\chi)_{L^{2}}=\int\pi_{\phi}(z)W(\psi
,\chi)(z^{\prime})d^{2n}z^{\prime}%
\end{equation}
for all $\psi,\chi\in\mathcal{S}(\mathbb{R}^{n})$ (see \textit{e.g.}
\cite{Birkbis}, \S 10.1.2). Lemma \ref{Lemma1} and the translational
covariance of the Wigner transform (\cite{Birkbis}, \S 9.2.2) imply that we
have
\[
\pi_{\phi}(z)(z^{\prime})=(2\pi\hbar)^{n}W(\widehat{T}(z)\phi)(z^{\prime
})=(2\pi\hbar)^{n}W\phi(z-z^{\prime})
\]
and hence
\begin{equation}
(\widehat{\Pi}_{\phi}(z)\psi|\chi)_{L^{2}}=(2\pi\hbar)^{n}\int W\phi
(z-z^{\prime})W(\psi,\chi)(z^{\prime})d^{2n}z^{\prime}~. \label{pio2}%
\end{equation}
Using definition\ (\ref{toeplitz1}) we thus have, by the Fubini--Tonnelli
theorem,
\begin{align*}
(\widehat{A}_{\phi}\psi|\chi)_{L^{2}}  &  =\int a(z)(\widehat{\Pi}_{\phi
}(z)\psi|\chi)_{L^{2}}d^{2n}z\\
&  =(2\pi\hbar)^{n}\int a(z)\left[  \int W\phi(z-z^{\prime})W(\psi
,\chi)(z^{\prime})d^{2n}z^{\prime}\right]  d^{2n}z\\
&  =(2\pi\hbar)^{n}\int\left[  \int a(z)W\phi(z-z^{\prime})d^{2n}z^{\prime
}\right]  W(\psi,\chi)(z^{\prime})d^{2n}z^{\prime}%
\end{align*}
hence the Weyl symbol of $\widehat{A}_{\phi}$ is $a=(2\pi\hbar)^{n}W\phi
\ast\mu$ as claimed.
\end{proof}

We next prove an extension of the usual symplectic covariance result
\cite{Folland,Birkbis,Littlejohn}
\begin{equation}
\operatorname*{Op}\nolimits_{\mathrm{W}}(a\circ S^{-1})=\widehat{S}%
\operatorname*{Op}\nolimits_{\mathrm{W}}(a)\widehat{S}^{-1} \label{sympco1}%
\end{equation}
for Weyl operators to the Toeplitz case. We recall the associated symplectic
covariance formulas (\cite{Birkbis}, \S 8.1.3 and 10.3.1)%
\begin{equation}
\widehat{S}\widehat{T}(z)\widehat{S}^{-1}=\widehat{T}(Sz)\text{ \ },\text{
\ }W(\widehat{S}\psi)(z)=W(\psi)(S^{-1}z)~. \label{sympco3}%
\end{equation}

\begin{corollary}
Let $\widehat{A}_{\phi}=\operatorname*{Op}\nolimits_{\phi}(a)$ and
$\widehat{S}\in\operatorname*{Mp}(n)$ have projection $S\in\operatorname*{Sp}%
(n)$. We have%
\begin{equation}
\operatorname*{Op}\nolimits_{\phi}(a\circ S^{-1})=\widehat{S}%
\operatorname*{Op}\nolimits_{\widehat{S}^{-1}\phi}(a)\widehat{S}^{-1}~.
\label{sympco2}%
\end{equation}

\end{corollary}

\begin{proof}
We begin by noticing that we have $\widehat{S}^{-1}\phi\in M^{1}%
(\mathbb{R}^{n})$ since $\phi\in M^{1}(\mathbb{R}^{n})$ hence $\widehat{S}%
^{-1}\phi$ is a \textit{bona fide} window. Since $W\phi(z-z^{\prime
})=W(\widehat{T}(z^{\prime})\phi)(z)$ we have, setting $z^{\prime\prime
}=S^{-1}z^{\prime}$ and using the relations (\ref{sympco3}),
\begin{align*}
(a\circ S^{-1})\ast W\phi(z)  &  =\int a(S^{-1}z^{\prime})W(\widehat{T}%
(z^{\prime})\phi)(z)d^{2n}z^{\prime}\\
&  =\int a(z^{\prime\prime})W(\widehat{T}(Sz^{\prime\prime})\phi
)(z)d^{2n}z^{\prime\prime}\\
&  =\int a(z^{\prime\prime})W(\widehat{S}\widehat{T}(z^{\prime\prime
})\widehat{S}^{-1}\phi)(z)d^{2n}z^{\prime\prime}\\
&  =\int a(z^{\prime\prime})W(\widehat{S}\widehat{T}(z^{\prime\prime
})\widehat{S}^{-1}\phi)(z)d^{2n}z^{\prime\prime}\\
&  =\int a(z^{\prime\prime})W(\widehat{T}(z^{\prime\prime})\widehat{S}%
^{-1}\phi)(S^{-1}z)d^{2n}z^{\prime\prime}\\
&  =(a\ast W(\widehat{S}^{-1}\phi)(S^{-1}z)
\end{align*}
It follows, using the covariance formula (\ref{sympco1}) for Weyl operators
that%
\[
\operatorname*{Op}\nolimits_{\mathrm{W}}\left[  (a\circ S^{-1})\ast
W\phi\right]  =\widehat{S}\operatorname*{Op}\nolimits_{\mathrm{W}}\left[
(a\ast W(\widehat{S}^{-1}\phi)\right]  \widehat{S}^{-1}%
\]
which is precisely (\ref{sympco2}) in view of formula (\ref{toeplitz3}) in
Proposition \ref{Prop3}.
\end{proof}

\section{Toeplitz Quantum States}

The following boundedness result was proven by Gr\"{o}chenig in \cite{Gro96},
Thm. 3; it is a particular case of Thm. 3.1 in Cordero and Gr\"{o}chenig
\cite{cogr0}:

\begin{proposition}
\label{Prop4}Let $a\in M^{1}(\mathbb{R}^{2n})$, then $\widehat{A}%
=\operatorname*{Op}\nolimits_{\mathrm{W}}(a)$ is of trace class:
$\widehat{A}\in\mathcal{L}_{1}(L^{2}(\mathbb{R}^{n}))$ and $||\widehat{A}%
||_{\mathcal{L}_{1}}\leq C||a||_{M^{1}}$ for some $C>0$ ($||\cdot
||_{\mathcal{L}_{1}}$ the trace norm and $||\cdot||_{M^{1}}$ a norm on
$M^{1}(\mathbb{R}^{2n})$).
\end{proposition}

Our main result makes use of the class of symbols $M^{1,\infty}(\mathbb{R}%
^{2n})$, which is a particular modulation space containing $L^{1}%
(\mathbb{R}^{2n})$. It is defined as follows \cite{Gro}: $a\in M^{1,\infty
}(\mathbb{R}^{2n})$ if and only if for some (and hence every) $b\in
\mathcal{S}(\mathbb{R}^{2n})$%
\begin{equation}
\sup_{\zeta\in\mathbb{R}^{2n}}|W^{2n}(a,b)(z,\zeta)|\in L^{1}(\mathbb{R}%
_{z}^{n})~. \tag{A8}\label{sjdual1}%
\end{equation}
When $b$ describes $\mathcal{S}(\mathbb{R}^{2n})$ the mappings $a\longmapsto
||a||_{b}$ defined by
\begin{subequations}
\begin{equation}
||a||_{b}=\int\sup_{\zeta\in\mathbb{R}^{2n}}|W^{2n}(a,b)(z,\zeta)|d^{2n}z
\label{ab}%
\end{equation}
is a family of equivalent norms on $M^{1,\infty}(\mathbb{R}^{2n})$ which
becomes a Banach space for the topology they define.

We have the following important convolution property between the symbol class
$M^{1,\infty}(\mathbb{R}^{2n})$ and the Feichtinger algebra $M^{1}%
(\mathbb{R}^{2n})$:%
\end{subequations}
\begin{equation}
M^{1,\infty}(\mathbb{R}^{2n})\ast M^{1}(\mathbb{R}^{2n})\subset M^{1}%
(\mathbb{R}^{2n}) \label{convincl}%
\end{equation}
(Prop. 2.4 in \cite{cogr0}).

Let us state and prove our main result. We assume that $\mu$ is a Borel
probability density function on $\mathbb{R}^{2n}$ with respect to the Lebesgue
measure: $\mu\in L^{1}(\mathbb{R}^{2n})$ and
\begin{equation}
\mu(z)\geq0\text{ \ },\text{ \ }\int\mu(z)d^{2n}z=1~. \label{muz}%
\end{equation}

\begin{theorem}
Let $\mu\in M^{1,\infty}(\mathbb{R}^{2n})$ and $\phi\in M^{1}(\mathbb{R}^{n}%
)$. Then%
\[
\widehat{\rho}=(2\pi\hbar)^{n}\operatorname*{Op}\nolimits_{\phi}(\mu
)=(2\pi\hbar)^{n}\operatorname*{Op}\nolimits_{\mathrm{W}}(\mu\ast W\phi)
\]
is a density operator, and there exists $C>0$ such that%
\[
||\widehat{\rho}||_{\mathcal{L}_{1}}\leq C||\mu\ast W\phi||_{M^{1,\infty}%
}||\phi||_{M^{1}}^{2}~.
\]

\end{theorem}

\begin{proof}
Let us prove that $\widehat{\rho}\in\mathcal{L}_{1}(L^{2}(\mathbb{R}^{n}))$.
In view of Proposition \ref{Prop4} it is sufficient to show that the Weyl
symbol $a=(2\pi\hbar)^{n}(\mu\ast W\phi)$ is in $M^{1}(\mathbb{R}^{2n})$. We
begin by observing that the condition $\phi\in M^{1}(\mathbb{R}^{n})$ implies
that $W\phi\in M^{1}(\mathbb{R}^{2n})$; this property is in fact a consequence
of the more general result Prop.2.5 in \cite{cogr0} but we give here a direct
independent proof. Let $b\in\mathcal{S}(\mathbb{R}^{2n})$; denoting by
$W^{2n}$ the cross-Wigner transform on $\mathbb{R}^{2n}$ we have (property
(\ref{marge}))%
\[
|\iint W^{2n}(W\phi,b)(z,\zeta)d^{2n}zd^{2n}\zeta|=|(W\phi|b)_{L^{2}%
(\mathbb{R}^{2n})}|<\infty
\]
and hence $W\phi\in M^{1}(\mathbb{R}^{2n})$ as claimed. In view of the
convolution property (\ref{convincl}) we have $\mu\ast W\phi\in M^{1}%
(\mathbb{R}^{2n})$ as desired. Let us show that $\widehat{\rho}\geq0$ (and
hence $\widehat{\rho}^{\ast}=\widehat{\rho}$) for all $\psi\in L^{2}%
(\mathbb{R}^{n})$. We have
\[
(\widehat{\rho}\psi|\psi)_{L^{2}}=\int\mu(z)(\widehat{\Pi}_{\phi}(z)\psi
|\psi)_{L^{2}}d^{2n}z~;
\]
since by definition
\[
\widehat{\Pi}_{\phi}\psi=(\psi|\widehat{T}(z)\phi)_{L^{2}}\widehat{T}(z)\phi
\]
we get%
\[
(\widehat{\Pi}_{\phi}(z)\psi|\psi)_{L^{2}}=|(\psi|\widehat{T}(z)\phi)|^{2}%
\geq0~.
\]
Let us finally prove that $\operatorname*{Tr}(\widehat{\rho})=1$. We have
$a\in M^{1}(\mathbb{R}^{2n})\subset L^{1}(\mathbb{R}^{2n})$ hence
(\cite{duwong}; also \cite{Birkbis} \S 12.3.2):
\[
\operatorname*{Tr}(\widehat{\rho})=(2\pi\hbar)^{-n}\int a(z)d^{2n}z=a_{\sigma
}(0)~.
\]
But we have
\[
a_{\sigma}=F_{\sigma}(W\phi\ast\mu)=(2\pi\hbar)^{2n}(F_{\sigma}W\phi
)(F_{\sigma}\mu)
\]
and hence, since $||\phi||_{L^{2}}=1$,%
\[
a_{\sigma}(0)=\int W\phi(z)d^{2n}z\int\mu(z)d^{2n}z=1
\]
that is $\operatorname*{Tr}(\widehat{\rho})=1$.
\end{proof}

A typical example is provided by anti--Wick quantization. Choose for window
$\phi$ the standard coherent state $\phi_{0}$:%
\begin{equation}
\phi_{0}(x)=(\pi\hbar)^{-n/4}e^{-|x|^{2}/2\hbar}~. \label{jstand}%
\end{equation}
Its Wigner transform is given by \cite{Bast,Birkbis,Wigner,Littlejohn}
\begin{equation}
W\phi_{0}(z)=(\pi\hbar)^{-n}e^{-|z|^{2}/\hbar} \label{wjstand}%
\end{equation}
and is thus a classical probability density
\begin{equation}
W\phi_{0}(z)>0\text{ \ },\text{ \ }\int W\phi_{0}(z)d^{n}z=1~. \label{wifobis}%
\end{equation}
It follows that $\mu\ast\phi$ is itself a probability density and that
$\operatorname*{Tr}(\widehat{\rho})=1$ (for an alternative proof see Boggiatto
and Cordero \cite{boco1}, Thm. 2.4).

\begin{acknowledgement}
This work was written while the author was holding the Giovanni Prodi visiting
professorship at the Julius-Maximilians-Universit\"{a}t W\"{u}rzburg during
the summer semester 2019.
\end{acknowledgement}

\end{document}